\documentclass[journal]{IEEEtran}

\usepackage{subfigure}
\usepackage{graphicx}
\usepackage{dsfont}
\usepackage{amsmath,amssymb,amsfonts}
\usepackage{lipsum}
\usepackage{cuted}
\usepackage{pgfplots}
\pgfplotsset{compat=1.18}
\usepackage{lipsum}
\usepackage[percent]{overpic}
\usepackage{xcolor}
\usepackage{comment}
\usepackage{booktabs}
\usepackage{algorithmic}
\usepackage{amsthm}
\def\BibTeX{{\rm B\kern-.05em{\sc i\kern-.025em b}\kern-.08em
    T\kern-.1667em\lower.7ex\hbox{E}\kern-.125emX}}
\usepackage{flushend}

\usepackage[ruled]{algorithm2e}

\newtheorem{thm}{Theorem}[section]

\newtheorem{rem}{Remark}

\def \T {^{\rm T}}
\def \b {{\rm b}}

\title{Optimal Derivative Feedback Control for an Active Magnetic Levitation System: An Experimental Study on Data-Driven Approaches
}

\author{
Saber Omidi $^{1}$,
Rene Akupan Ebunle $^{2}$,  
and Se Young Yoon $^{3}$ 
\thanks{$^{1}$ Saber Omidi, Department of Mechanical Engineering,
               University of New Hampshire, Durham, NH 03824, USA.
        {\tt\small saber.omidi@unh.edu}}%
\thanks{$^{2}$ Rene Akupan Ebunle,
               Department of Electrical and Computer Engineering, 
               University of New Hampshire, Durham, NH 03824, USA.
        {\tt\small ear1080@usnh.edu}}%
\thanks{$^{3}$ Se Young Yoon, Department of Electrical and Computer Engineering, 
               University of New Hampshire, Durham, NH 03824, USA.
        {\tt\small seyoung.yoon@unh.edu}}        
}

\begin{document}

\maketitle

\begin{center}
\small\textit{Preprint. This manuscript is under review and has not been peer reviewed.}
\end{center}
\vspace{-8pt}

\begin{abstract}

This paper presents the design and implementation of data-driven optimal derivative feedback controllers for an active magnetic levitation system. A direct, model-free control design method based on the reinforcement learning framework is compared with an indirect optimal control design derived from a numerically identified mathematical model of the system. For the direct model-free approach, a policy iteration procedure is proposed, which adds an iteration layer called the epoch loop to gather multiple sets of process data, providing a more diverse dataset and helping reduce learning biases. This direct control design method is evaluated against a comparable optimal control solution designed from a plant model obtained through the combined Dynamic Mode Decomposition with Control (DMDc) and Prediction Error Minimization (PEM) system identification. Results show that while both controllers can stabilize and improve the performance of the magnetic levitation system when compared to controllers designed from a nominal model, the direct model-free approach consistently outperforms the indirect solution when multiple epochs are allowed. The iterative refinement of the optimal control law over the epoch loop provides the direct approach a clear advantage over the indirect method, which relies on a single set of system data to determine the identified model and control. 
\end{abstract}

\begin{IEEEkeywords}
Derivative feedback control, reinforcement learning, policy iteration, magnetic levitation system, optimal control, system identification, dynamic mode decomposition
\end{IEEEkeywords}


 \section{Introduction}\label{sec:intro}

\IEEEPARstart{M}{agnetic} \textcolor{black}{levitation systems are electromechanical systems that employ magnetic forces to minimize contact and eliminate frictional losses during the relative motion of multiple surfaces \cite{wang2017high}. {\it Active Magnetic Levitation} (AML) systems have found use in high bandwidth applications that require rapid response, including magnetic bearings, vibration-isolation devices, and high-speed maglev transportation \cite{samanta2008magnetic, tsuda2009vibration, yoon2012control}. AML systems have complex nonlinear dynamics, where electromechanical and electromagnetic components interact resulting in high sensitivity to parameter variations, and inherent open-loop instability \cite{berkelman2011magnetic, chamraz2021stabilization}. For example, small physical discrepancies in actuator geometry or misalignment between the actuation and sensing axes can introduce substantial uncertainty in the calculation of magnetic forces, and complicate the determination of the magnetic equilibrium for stable levitation \cite{arthur2020robust}. These uncertainties limit the effectiveness of the common linearization and model-based linear controller design solutions \cite{khimani2017implementation, yaseen2017comparative, sarmad2016sampled},  as the linearization and the resulting nominal model depend on access to accurate information of the equilibrium states. Uncertainty in the equilibrium point leads to discrepancies between the sensing center and the actuation center, resulting in steady-state errors, increased actuator effort, and possibly instability due to the model uncertainty.}

\textcolor{black}{The {\it Derivative Feedback Control} (DFC) framework offers a robust solution for addressing the challenges associated to uncertainty in the equilibrium states or unknown bias in the feedback measurement \cite{cardim2008control}. DFC uses the derivative of the state or output as the feedback measurement for control; a stabilizing DFC naturally drives the system to its actual equilibrium, even if there is a bias in the feedback measurement, because the system's state derivative is zero at the equilibrium state \cite{zaheer2021derivative}. This property makes DFC particularly suitable for systems such as AML, where frequent misalignment between the sensing and actuation axes results in such measurement bias. One commonly referred limitation of the DFC is its sensitivity to high-frequency disturbances due to the derivative calculation. In many cases, the derivatives of the states can be measured directly by sensors such as accelerometers, gyroscopes or other differential sensors. When such direct measurements are not available, low-pass filtering and numerical techniques for differentiation of noisy state measurements \cite{verma:derivative} can mitigate the effect of high-frequency noise above the bandwidth of the system \cite{arthur2020robust}. Another limitation is that DFC by itself does not address the challenges of dynamic model uncertainty and robust control, opening the door for data-driven variations that can provide a robust control implementation for systems with uncertain or unknown dynamics.}

\textcolor{black}{{\it Reinforcement Learning} (RL) strategies provide promising data-driven solutions to synthesize optimal control laws when access to an accurate mathematical model is restricted, by directly interacting with the plant to be controlled \cite{10536605,10388382}. Two prominent data-driven algorithms within this class of solutions are the {\it Policy Iteration} (PI) and the {\it Value Iteration} (VI) \cite{lewis_PI, kleinman1968, puterman1978modified, farahmand2010error, bertsekas2019reinforcement}, which iterate over the derived policy/control and its corresponding value function to determine the optimal solution. The VI algorithm is preferred when an initial stable control policy is not available to initiate iterations, whereas the PI algorithm requires the plant to be stabilized with an initial control policy. Such advantage offers limited benefits in many applications where unstable operation during the training phase is impractical or unsafe. Vrabie et al. \cite{vrabie2009adaptive} proposed a state-feedback solution to the LQR problem with partial knowledge of the system dynamics within the PI framework, and Jiang and Jian \cite{jiang} extended this to a complete model-free solution. These frameworks were later extended to the output feedback control in \cite{modares2016optimal,rizvi2019reinforcement}. Other significant extensions include a Q-learning approach for discrete-time $H_{\infty}$ control \cite{rizvi2018output}, DFC \cite{zaheer2023derivative}, and output-difference feedback control (ODFC) under measurement and process noise \cite{zaheer2024model}.}

\textcolor{black}{
Data-driven system identification can also offer a feasible solution for designing control laws including DFC in applications with uncertain or unknown system models. System identification can also provide mathematical models suitable for stability and robustness analyses prior to  the implementation of the closed-loop control. Non-parametric subspace methods, such as {\it Dynamic Mode Decomposition} (DMD), are recognized for their computational efficiency \cite{dey2022dynamic, Baddoo2023Physics}. In contrast, parametric techniques such as {\it Prediction Error Minimization} (PEM) generally achieve higher model accuracy but are more computationally intensive and sensitive to factors such as model structure and initial conditions \cite{Cheon2009Anew, Jing2020Identification, Wei2024Parameter}. For systems subject to external inputs and disturbances, specialized techniques exist for identifying input-output dynamics, including {\it DMD with Control} (DMDc) \cite{Nedzhibov2023}.}

\textcolor{black}{
Despite recent advances in DFC and data-driven control, much of the existing literature investigates idealized problems and numerical simulations, and overlooks many challenges related to practical implementation in physical systems, such as AML systems with unstable dynamics, measurement noise, model uncertainties, and uncertainty in the equilibrium state. In this work we address this knowledge gap by investigating the effectiveness of data-driven DFC control for an AML system, where DFC is specifically selected to address the uncertainty in the magnetic equilibrium and the resulting bias in the measurement. In particular, we offer a direct model-free solution where the optimal DFC is synthesized using a modified PI algorithm, and an indirect alternative where a combined DMDc-PEM technique first provides an identified mathematical model to design the DFC control. The primary contributions of this paper are as follows:
\begin{enumerate}
    \item A practical multi-epoch training procedure is introduced for a model-free PI algorithm. The proposed multi-epoch approach incorporates an outer-loop that iteratively re-collects data and performs PI using the policies from previous iterations. This process is designed to reduce training bias and increase the diversity of training samples, thereby supporting robust solution convergence in experimental settings.
    \item A rigorous experimental study is conducted between the solution of the proposed multi-epoch model-free PI approach and an indirect data-driven optimal DFC approach. This indirect approach combines data-driven system identification and model-based DFC design. A detailed comparative analysis is conducted using these data-driven control strategies to quantify improvements in performance and robustness. 
\end{enumerate}}

The remainder of this paper is organized as follows: Sect.~\ref{sec:ProblemStatement} introduces the optimal DFC problem, the model-based PI solution, and a model-free PI algorithm for synthesizing the data-driven DFC. Section~\ref{sec:SID} presents a system identification method combining DMDc and PEM that is used in this study to derive a comparable indirect data-driven DFC for the experimental study. The data-driven DFC solutions investigated in this study are implemented and tested on a AML system in Sect.~\ref{sec:ResultAndDiscussion}. The mathematical model of the AML system is presented in this section, for which the data-driven control solutions are studied through numerical simulations and experimental tests. Finally, concluding notes are shared in Sect.~\ref{sec:Conclusion}.

\section{Data-Driven Derivative Feedback Control Through Policy Iteration} \label{sec:ProblemStatement}
\textcolor{black}{
The notation employed throughout this paper is defined as follows. $\mathbb{R}^{n}$ denotes the $n$-dimensional Euclidean space, and $\mathbb{R}^{n \times m}$ represents the set of all real matrices of dimension $n \times m$. For a matrix $\boldsymbol{M}$, $\boldsymbol{M}^T$ denotes its transpose and $\boldsymbol{M}^{-1}$ denotes its inverse if it exists. If $\boldsymbol{M}$ is symmetric, the notation $\boldsymbol{M} \ge 0$ ($\boldsymbol{M} > 0$) indicates that it is positive semi-definite (positive definite). The Kronecker product of matrices $\boldsymbol{A}$ and $\boldsymbol{B}$ is denoted by $\boldsymbol{A} \otimes \boldsymbol{B}$. The vectorization operator $\text{vec}(\boldsymbol{M})$ stacks the columns of matrix $\boldsymbol{M}$ into a single column vector. $\boldsymbol{I}_n$ denotes the $n \times n$ identity matrix. For a vector $\boldsymbol{x}(t)$, $\dot{\boldsymbol{x}}(t)$ denotes its time derivative. The notation $\langle \cdot, \cdot \rangle$ refers to the standard inner product operator in Euclidean space.}

\subsection{Derivative Feedback LQR Control for LTI System}
Consider the following system,
\begin{equation} \label{LTI system}
        \boldsymbol{\dot{x}}(t) = \boldsymbol{A}\boldsymbol{x}(t) + \boldsymbol{B}\boldsymbol{u}(t), 
        \quad \boldsymbol{x}(t_0) = \boldsymbol{x}_0,
    \end{equation}
where \( \boldsymbol{x}(t) \in \mathbb{R}^n \) is the state with initial value \(\boldsymbol{x}(t_0)=\boldsymbol{x}_0\),
and \( \boldsymbol{u}(t) \in \mathbb{R}^m \) 
is the control input. Matrices \( \boldsymbol{A} \in \mathbb{R}^{n \times n} \)
and \( \boldsymbol{B} \in \mathbb{R}^{n \times m} \) 
are the state and input matrices,
respectively. It is assumed that the system is controllable, and \( \boldsymbol{A} \) is full rank.
A state DFC law takes the form    
\begin{equation} \label{DFC gain}
\boldsymbol{u}(t)=-\boldsymbol{K} \boldsymbol{\dot{x}}(t),
\end{equation}
where $\boldsymbol{K} \in \mathbb{R}^{m \times n}$ is the controller gain, and the closed-loop system becomes
\begin{equation} \label{DFC closed loop}
(\boldsymbol{I} + \boldsymbol{BK})\boldsymbol{\dot{x}} = \boldsymbol{A}\boldsymbol{x}(t).
\end{equation}

For the quadratic cost function
\begin{equation}\label{DFC cost function}
        V(\boldsymbol{\dot{x}},\boldsymbol{u}) = \int_{t_0}^{\infty} \left(\boldsymbol{\dot{x}}^T(\tau) 
        \boldsymbol{Q} \boldsymbol{\dot{x}}(\tau) +
        \boldsymbol{u}^T(\tau) \boldsymbol{R} \boldsymbol{u}(\tau) \right) d\tau,
        \end{equation}
with \(\boldsymbol{Q} \geq 0  \in \mathbb{R}^{n \times n}\) and \(\boldsymbol{R} > 0  \in \mathbb{R}^{m \times m}\), the optimal state DFC gain $\boldsymbol{K}^{*}$ for \eqref{DFC gain} can be derived from the {\it Algebraic Riccati Equation} (ARE),
\begin{equation}\label{DFC ARE}
\boldsymbol{P}^*\!\boldsymbol{A}^{\!-1}\!+
\!{\boldsymbol{A}^{\!-1}}\T\!\boldsymbol{P}^*\!-
\!\boldsymbol{P}^*\!\boldsymbol{A}^{-1}\!\boldsymbol{B}
\!\boldsymbol{R}^{-1}\!\boldsymbol{B}\T\!{\boldsymbol{A}^{-1}}
\T\!\boldsymbol{P}^*\!+\!\boldsymbol{Q}\!=\!\boldsymbol{0},
\end{equation}
where positive definite matrix $\boldsymbol{P}^*\in \mathbb{R}^{n \times n}$ is the solution
of the ARE and 
        \begin{equation}\label{DFC optimal gain}
        \boldsymbol{K}^{*}=-\boldsymbol{R}^{-1}\boldsymbol{B}
        \T {\boldsymbol{A}^{-1}}\T \boldsymbol{P}^{*},
        \end{equation} 
\cite{abdelaziz05}. Under the control gain \eqref{DFC optimal gain},
the optimal value of \eqref{DFC cost function} takes the form
\begin{equation}\label{finite value DFC cost functiont}
{V}( \boldsymbol{x}(t_0)) =   \boldsymbol{x}^T_0 \boldsymbol{P}^*\boldsymbol{x}_0.
\end{equation}  
\textcolor{black}{
The solution to \eqref{DFC ARE} exists if the pair ($\boldsymbol{A}$, $\boldsymbol{B}$) is stabilizable, ($\boldsymbol{Q}^{1/2}$, $\boldsymbol{A}$) is observable, and $\boldsymbol{A}$ is nonsingular \cite{abdelaziz05}.}

\subsection{Model-Based Policy Iteration (PI)}\label{model based}
A PI algorithm is introduced next for solving \eqref{DFC ARE}-\eqref{DFC optimal gain}, which is to become the foundation for the proposed model-free PI framework for data-driven optimal control.

\begin{thm} \label{thm_1}
Let $ \boldsymbol{K}_1$ to be any stabilizing state DFC gain for the system \eqref{LTI system} and let $\boldsymbol{P}_i>\boldsymbol{0}$ represent the unique positive definite solution to the equation
\begin{equation}\label{p_iteration}
\boldsymbol{P}_i \boldsymbol{A}_{i}^{-1}+{\boldsymbol{A}_{i}^{-1}}\T \boldsymbol{P}_i+\boldsymbol{Q}
+\boldsymbol{K}_i\T \boldsymbol{R} \boldsymbol{K}_i=\boldsymbol{0},
\end{equation}
where $\boldsymbol{A}_{i}^{-1}=\boldsymbol{A}^{-1}(\boldsymbol{I}+\boldsymbol{BK}_{i})$. For
$\boldsymbol{K}_{i+1}$ defined as 
\begin{equation}\label{k_iteration}
\boldsymbol{K}_{i+1}=-\boldsymbol{R}^{-1}\boldsymbol{B}^T {\boldsymbol{A}_{i}^{-1}}^T \boldsymbol{P}_{i},
\end{equation}
for $i=0,1,2,\ldots,$ the following properties hold:
\begin{enumerate}
  \item $(\boldsymbol{I}+\boldsymbol{B}\boldsymbol{K}_{i+1})^{-1}\boldsymbol{A}$ exists and is Hurwitz,
  \item $\boldsymbol{P}^*\leq \boldsymbol{P}_{i+1}\leq \boldsymbol{P}_i$,
  \item $\lim_{i \to \infty}\boldsymbol{P}_i=\boldsymbol{P}^*$, $\lim_{i \to \infty}\boldsymbol{K}_i=\boldsymbol{K}^*$.
\end{enumerate}
\end{thm}

{\color{black}
\begin{proof}
Detailed proof of this theorem is found in \cite{zaheer2023derivative}, and here we provide sketch proof for completeness.  

Property (1) states that the closed-loop system under the DFC gain $\boldsymbol{K}_{i+1}$ is stable if $\boldsymbol{K}_{i}$ is a stabilizing gain. This can be proved by considering the Lyapunov function ${V_i}( \boldsymbol{x}(t)) = \boldsymbol{x}^T (t) \boldsymbol{P}_i \boldsymbol{x}(t)$. The time derivative of $V_i$ along the trajectories of the states under the control $\boldsymbol{K}_{i+1}$ can be shown to be negative semi-definite from \eqref{p_iteration}-\eqref{k_iteration}, and the LaSalle's invariance principle completes the stability proof.

Property (2) states the monotonic convergence the matrix $\boldsymbol{P}_i$ to $\boldsymbol{P}^*$. This can be proven by considering the difference of \eqref{p_iteration} for $i$ and $i+1$, which reduces to another ARE with positive semi-definite matrix $(\boldsymbol{P}_i-\boldsymbol{P}_{i+1})\geq 0$. The existence of this positive semi-definite matrix is guaranteed by the stability of the closed-loop system under both $\boldsymbol{K}_i$ and $\boldsymbol{K}_{i+1}$. 

Property (3) states that the sequences of positive definite matrices $\boldsymbol{P}_i$ and DFC gain $\boldsymbol{K}_i$ asymptotically approach their optimal values. This is proven with the help of the Monotone Convergence Theorem, and by observing that the limits of \eqref{p_iteration} and \eqref{k_iteration} satisfy  \eqref{DFC ARE} and \eqref{DFC optimal gain}, respectively. This concludes the proof.
\end{proof}}

\textcolor{black}{
Theorem \ref{thm_1} establishes theoretical stability and convergence guarantees for the PI framework. Furthermore, for a cost function 
\begin{equation}\label{cost_iteration}
V_{i}(\boldsymbol{\dot{x}},\boldsymbol{u},t) = \int_{t}^{\infty} \boldsymbol{\dot{x}}^T(\tau)(\boldsymbol{Q}+
\boldsymbol{K}_{i}^T \boldsymbol{RK}_{i}) \boldsymbol{\dot{x}}(\tau) d\tau,
\end{equation}
under the DFC gain $\boldsymbol{K}_i$, it comes from \eqref{p_iteration} that the corresponding cost value is
\begin{equation}\label{finite_value_iteration}
{V_i}( \boldsymbol{x}(t)) = \boldsymbol{x}^T (t) \boldsymbol{P}_i \boldsymbol{x}(t).
\end{equation}
A data-driven implementation of the results of Theorem~\ref{thm_1} is presented next such that it achieves equivalent convergence without knowledge of the system matrices $\boldsymbol{A}$ and $\boldsymbol{B}$.}

\subsection{Model-free Policy Iteration}\label{model free}
A model-free PI algorithm can be derived by extending the iterative calculations presented in Theorem \ref{thm_1}. Specifically, the goal is to solve \eqref{p_iteration} and \eqref{k_iteration} without knowledge of the system matrices, but instead leveraging measurements of the input and the state.
 
For system \eqref{LTI system} under the DFC \eqref{DFC gain}, it comes from the derivation in \cite{jiang} and \cite{zaheer2023derivative} that Eqs.~\eqref{k_iteration}, \eqref{cost_iteration}, and \eqref{finite_value_iteration}  yield
\begin{multline}\label{equal_cost_equations}
 \boldsymbol{x}^T(t)\boldsymbol{P}_i  \boldsymbol{x}(t)- \boldsymbol{x}^T(t+T)\boldsymbol{P}_i \boldsymbol{x}(t+T) =\\
\int_{t}^{t+T} \dot{\boldsymbol{x}}^T(\tau)(\boldsymbol{Q}+\boldsymbol{K}_i^T \boldsymbol{RK}_i) \dot{\boldsymbol{x}}(\tau) d\tau\\
-2\int_{t}^{t+T}\left( \boldsymbol{u}(\tau)+\boldsymbol{K}_i \dot{\boldsymbol{x}}(\tau) \right) \T 
\boldsymbol{RK}_{i+1}   \dot{\boldsymbol{x}}(\tau)d\tau,
\end{multline}
for $T>0$. The first term in the right-hand side of \eqref{equal_cost_equations} can be written as 
\begin{multline}\label{right_equal_cost_equations}
\lefteqn{\int_{t}^{t+T} \boldsymbol{\dot{x}}\T(\tau)(\boldsymbol{Q}+\boldsymbol{K}_i\T \boldsymbol{RK}_i) \boldsymbol{\dot{x}}(\tau) d\tau} \\
= \boldsymbol{I}_{xx}^{t,t+T}\text{vec}(\boldsymbol{Q}+\boldsymbol{K}_i\T \boldsymbol{RK}_i),
\end{multline}
while the second integral is
\begin{multline}\label{integrations_equal_cost_equations}
\lefteqn{\int_{t}^{t+T} \left( \boldsymbol{u}(\tau)+\boldsymbol{K}_i  \boldsymbol{\dot{x}}
(\tau) \right) \T \boldsymbol{RK}_{i+1}  \boldsymbol{\dot{x}}(\tau)d\tau }\\
=\left( \boldsymbol{I}_{xx}^{t,t+T}(\boldsymbol{I}_n \otimes \boldsymbol{K}_i\T \boldsymbol{R})\right. \\ \left. +\boldsymbol{I}_{xu}^{t,t+T} (\boldsymbol{I}_n \otimes \boldsymbol{R}) \right) \text{vec}(\boldsymbol{K}_{i+1}),
\end{multline}
where
\begin{align*}
\boldsymbol{I}_{xx}^{t,t+T}&=\int_{t}^{t+T} \boldsymbol{\dot{x}}
\T(\tau)\otimes  \boldsymbol{\dot{x}}\T(\tau) d\tau,\\
\boldsymbol{I}_{xu}^{t,t+T}&=\int_{t}^{t+T} \boldsymbol{\dot{x}}\T(\tau)\otimes \boldsymbol{u}\T(\tau) d\tau.
\end{align*}

Let $\bar{\boldsymbol{x}} = \boldsymbol{x} + \boldsymbol{x}_\b$ be a biased measurement of the system state $\boldsymbol{x}$, where $\boldsymbol{x}_\b$ represents an unknown static bias. Then, the true state is $\boldsymbol{x} = \boldsymbol{\bar{x}}-\boldsymbol{x}_{\b}$ and,
\begin{equation*}\
\boldsymbol{x}\T(t)\boldsymbol{P}_i \boldsymbol{x}(t)=\boldsymbol{\bar{x}}\T(t)\boldsymbol{P}_i \boldsymbol{\bar{x}}(t)+\epsilon\T \boldsymbol{\bar{x}}+\boldsymbol{x}_\b\T \boldsymbol{P}_i \boldsymbol{x}_\b,
\end{equation*}
for $\boldsymbol{\epsilon}=-2\boldsymbol{P}_i \boldsymbol{x}_\b$. Accordingly, the left-hand side of 
\eqref{equal_cost_equations} is
\begin{align}\label{repre_left_2}
&\boldsymbol{x}\T(t)\boldsymbol{P}_i \boldsymbol{x}(t)-\boldsymbol{x}\T(t+T)\boldsymbol{P}_i \boldsymbol{x}(t+T)\\ \nonumber =&\text{vec}(\boldsymbol{P}_i)\T \left( \boldsymbol{x}_\kappa(t+T)- \boldsymbol{x}_\kappa(t) \right)
+\boldsymbol{\epsilon}\T \left( \boldsymbol{\bar{x}}(t+T)- \boldsymbol{\bar{x}}(t) \right),
\end{align}
where $\boldsymbol{x}_\kappa$ represents a polynomial basis derived
from the measured system states, 
$\boldsymbol{x}_\kappa=\boldsymbol{\bar{x}} \otimes \boldsymbol{\bar{x}}$.
For $N$ data sample sets collected over a
time intervals $[(j-1)T,jT]$, where $j=1,2,\cdots,N$, 
Eqs.~\eqref{equal_cost_equations}-\eqref{repre_left_2} 
result in,
\begin{equation}\label{vectorized_form}
\boldsymbol{X}_i \begin{bmatrix}
\text{vec}(\boldsymbol{P}_i),
\boldsymbol{\epsilon},
\text{vec} (\boldsymbol{K}_{i+1})
\end{bmatrix}\T= \boldsymbol{Y}_i,
\end{equation}
where
\begin{align*}
\boldsymbol{X}_i&=\begin{bmatrix}\boldsymbol{\Delta}_{x_\kappa},\boldsymbol{\Delta}_x,-2\boldsymbol{I}_{xx}(\boldsymbol{I}_n \otimes \boldsymbol{K}_i\T \boldsymbol{R})-2\boldsymbol{I}_{xu}(\boldsymbol{I}_n \otimes \boldsymbol{R})\end{bmatrix},\\
\boldsymbol{\Delta}_{\boldsymbol{x}_\kappa}&=\begin{bmatrix}\boldsymbol{x}_\kappa(T)\!-\!\boldsymbol{x}_\kappa(0),~\hdots,~\boldsymbol{x}_\kappa(NT)\!-\!\boldsymbol{x}_\kappa ((N-1)T)\end{bmatrix}\T,\\
\boldsymbol{\Delta}_{x}&=\begin{bmatrix}\boldsymbol{\bar{x}}(T)- \boldsymbol{\bar{x}}(0), & \hdots, &  \boldsymbol{\bar{x}}(NT)-\boldsymbol{\bar{x}} ((N-1)T)\end{bmatrix}\T,\\
\boldsymbol{Y}_i&=-\boldsymbol{I}_{xx}\text{vec}(\boldsymbol{Q}+\boldsymbol{K}_i\T \boldsymbol{RK}_i),\\
\boldsymbol{I}_{xx}&=\begin{bmatrix}\boldsymbol{I}_{xx}^{0,T}, & \boldsymbol{I}_{xx}^{T,2T}, & \hdots, & \boldsymbol{I}_{xx}^{(N-1)T,NT}\end{bmatrix}\T,\\
\boldsymbol{I}_{xu}&=\begin{bmatrix}\boldsymbol{I}_{xu}^{0,T}, & \boldsymbol{I}_{xu}^{T,2T}, & \hdots, & \boldsymbol{I}_{xu}^{(N-1)T,NT}\end{bmatrix}\T,
\end{align*}
and \(\boldsymbol{I}_n\) is the identity matrix with dimension $n\times n$. 
Assuming that $\boldsymbol{X}_i$ is a full column rank, a solution to \eqref{vectorized_form} exists as 
\begin{equation}\label{inv_vectorized_form}
\begin{bmatrix}
\text{vec}(\hat{\boldsymbol{P}}_i), \hat{\boldsymbol{\epsilon}},
\text{vec} (\hat{\boldsymbol{K}}_{i+1})
\end{bmatrix}\T= (\boldsymbol{X}_i\T
\boldsymbol{X}_i)^{-1}\boldsymbol{X}_i\T  \boldsymbol{Y}_i,
\end{equation}
which is equivalent to the solution to \eqref{p_iteration}
and \eqref{k_iteration}. Moreover, if $\boldsymbol{K}_1$ 
is a stabilizing gain for \eqref{DFC closed loop},
then  Theorem \ref{thm_1} implies that $\hat{\boldsymbol{P}}_i$ and $\hat{\boldsymbol{K}}_{i+1}$
must converge to $\boldsymbol{P}^*$ and $\boldsymbol{K}^*$, respectively, 
as $i \to \infty$. This result is formalized in the following theorem. 

\begin{thm}
    \label{thm:convergence}
If the data matrix $\boldsymbol{X}_i$ is full column rank, then $\hat{\boldsymbol{P}}_i$ and $\hat{\boldsymbol{K}}_{i+1}$ obtained from \eqref{inv_vectorized_form} equal the solutions of \eqref{p_iteration} and \eqref{k_iteration}. Moreover, if the initial gain $\boldsymbol{K}_1$ stabilizes \eqref{DFC closed loop}, then the sequences $\hat{\boldsymbol{P}}_i$ and $\hat{\boldsymbol{K}}_{i}$ asymptotically converge to the optimal values $\boldsymbol{P}^*$ and $\boldsymbol{K}^*$ as $i\to\infty$.
\end{thm}

\begin{proof}
\textcolor{black}{
    Provided the data matrix $\boldsymbol{X}_i$ has full column rank, the solution of \eqref{inv_vectorized_form} produces $\hat{\boldsymbol{P}}_i$ and $\hat{\boldsymbol{K}}_{i+1}$ that are equivalent to $\boldsymbol{P}_i$ and $\boldsymbol{K}_{i+1}$ in \eqref{vectorized_form}, and to the solution to the model-based equations \eqref{p_iteration} and \eqref{k_iteration}. Therefore, the convergence guarantees established in Theorem \ref{thm_1} are directly applicable, ensuring that the data-driven estimates $\hat{\boldsymbol{P}}_i$ and $\hat{\boldsymbol{K}}_{i+1}$ converge to the true optimal values.}
\end{proof}

\textcolor{black}{
The implementation details of the proposed model-free PI calculation of \eqref{inv_vectorized_form} are presented in Algorithm~\ref{algorithm1}. The algorithm is structured with an outer $\kappa$ epoch loop that repeatedly runs the inner $i$ loop of the PI algorithm to improve robustness to noise and bias in the training data and to aid the PI algorithm in converging to the true optimal policy. This is achieved in each epoch iteration by repeating the collection of training data and the PI process with the policy/control from the previous epoch as the initial control gain. A threshold factor $\bar{\zeta}$ is introduced to evaluate the solution across subsequent epochs, and the learning process is terminated if the magnitude variation of the cost function value is below the prescribed $\bar{\zeta}$.}

\begin{algorithm}[!t] \label{algorithm1}
\caption{Multi-Epoch Model-Free Online PI}
\begin{algorithmic}[1]
\REQUIRE Initial stabilizing gain $\boldsymbol{K}_{1}$, inner convergence tolerance $\bar{\eta}$ and outer convergence tolerance $\bar{\zeta}$.
\ENSURE Optimal DFC gain $\boldsymbol{K}^{*}$.

\STATE Initialize epoch index $\kappa \leftarrow 1$.
\STATE Initialize control gain for the epoch $\boldsymbol{\mathcal{K}}_\kappa=\boldsymbol{K}_1$
\STATE Initialize cost deviation $\Delta V \leftarrow \infty$.

\WHILE{$\Delta V \geq \bar{\zeta}$}
    \STATE \textbf{Step 1: Data Collection}
    \STATE Apply control $\boldsymbol{u}(t)\! =\! -\boldsymbol{K}_{1}\boldsymbol{\dot{x}}(t)\!+\!\text{excitation noise}$.
    \STATE Collect samples of $\boldsymbol{x}$, $\dot{\boldsymbol{x}}$, and $\boldsymbol{u}$ in $[t, t+NT]$.

    \STATE \textbf{Step 2: Policy Evaluation and Improvement}
    \STATE Initialize inner iteration index $i \leftarrow 1$.
    \REPEAT
        \STATE Construct $\boldsymbol{X}_i$ and vectors $\boldsymbol{Y}_i$.
        \STATE Solve \eqref{inv_vectorized_form} for $\boldsymbol{P}_i$ and $\boldsymbol{K}_{i+1}$.
    \UNTIL{$||\boldsymbol{P}_i - \boldsymbol{P}_{i-1}|| < \bar{\eta}$ (for $i >1)$}

    \STATE \textbf{Step 3: Epoch Update}
    \STATE Calculate cost function using \eqref{cost_iteration} for $\boldsymbol{K}_{i}$.
    \STATE Increment epoch counter: $\kappa \leftarrow \kappa + 1$.
    \STATE Update control gain for next epoch: $\boldsymbol{\mathcal{K}}_{\kappa} \leftarrow \boldsymbol{K}_{i}$.
    \STATE Update $\Delta V = |V_{\kappa} - V_{\kappa-1}|$ (for $\kappa > 0$).
    \STATE $\boldsymbol{K}_1\leftarrow \boldsymbol{\mathcal{K}}_\kappa$
\ENDWHILE

\STATE \textbf{Return} $\boldsymbol{K}^{*} = \boldsymbol{\mathcal{K}}_{\kappa}$.
\end{algorithmic}
\end{algorithm}

\textcolor{black}{
The algorithm is initialized with a stabilizing control gain $\boldsymbol{K}_1$. According to Theorem~\ref{thm:convergence}, the closed-loop system remains stable throughout Algorithm~\ref{algorithm1}. As a result, each trained gain $\boldsymbol{K}_i$ calculated by the iteration loop and $\boldsymbol{\mathcal{K}_{\kappa}}$ by the epoch loop 
preserve the closed-loop stability. Furthermore, the gain $\boldsymbol{K}_i$ for each epoch satisfy the convergence condition of Theorem~\ref{thm_1}.}

\begin{rem}
\textcolor{black}{
The multi-epoch PI algorithm is developed to enhance the robustness of the standard PI algorithm in the presence of noise and bias within training data. This increased robustness results in higher computational complexity, primarily due to the least-squares estimation step in the inner PI loop (Step 2). The computational complexity of this step is $O(L^3)$, where $L = \frac{n(n+1)}{2} + n + mn$ denotes the total number of unknown parameters in $\boldsymbol{P}_i$, $\epsilon$, and $\boldsymbol{K}_{i+1}$ \cite{jiang}. The addition of the outer epoch loop causes this process to repeat for a specified maximum number of epochs, resulting in total complexity that scales linearly with the number of epochs. Consequently, the method maintains a polynomial complexity class, which supports computational efficiency and suitability for implementation on standard processors.}
\end{rem}

\section{Data-Driven Derivative Feedback Control Through System Identification}\label{sec:SID}
\textcolor{black}{In order to carry out a practical comparison test to evaluate the proposed model-free PI solution, an indirect approach is considered in which an optimal DFC is designed from an identified model of \eqref{LTI system}. Here we present a composite system identification technique, which integrates DMDc for initial model estimation and PEM method for parameter refinement based on measured states and input signals \cite{Nedzhibov2023, Joshua2016}.} 

\subsection{Dynamic Mode Decomposition with Control}
{\color{black} From \(N\) observations of the state $\boldsymbol{x}(t_i)$ and input $\boldsymbol{u}(t_i)$, let $\boldsymbol{\chi}=[\boldsymbol{x}(t_1)~\boldsymbol{x}(t_2)~\cdots~\boldsymbol{x}(t_N)]\in\mathbb{R}^{n\times N}$ and $\boldsymbol{\Upsilon}=[\boldsymbol{u}(t_1)~\boldsymbol{u}(t_2)~\cdots~\boldsymbol{u}(t_N)]\in\mathbb{R}^{m\times N}$. It yields from \eqref{LTI system} that
\begin{equation}
   \boldsymbol{\dot{\chi}} = 
   \begin{bmatrix}
       \boldsymbol{A} & \boldsymbol{B}
   \end{bmatrix}
   \begin{bmatrix}
       \boldsymbol{\chi}\T & \boldsymbol{\Upsilon}\T
   \end{bmatrix}\T.
\end{equation} 
For sampled data matrices \( \boldsymbol{\dot{\chi}} \) and \( \begin{bmatrix} \boldsymbol{\chi}\T & \boldsymbol{\Upsilon}\T \end{bmatrix}\) collected from the plant to be identified, DMDc is performed to determine the best-fit mappings of \( \boldsymbol{A}\) and \( \boldsymbol{B}\). If we let
\begin{equation} \label{eq:io_mapping}
    \boldsymbol{\dot{\chi}} = \boldsymbol{\Gamma \Phi},
\end{equation}
where \( \boldsymbol{\Gamma}  =\boldsymbol{ \begin{bmatrix} A & B \end{bmatrix}}\) and \( \boldsymbol{\Phi}\T = \begin{bmatrix} \boldsymbol{\chi}\T & \boldsymbol{\Upsilon}\T
\end{bmatrix} \), it follows that
\begin{equation}\label{eq:Gamma_est}
   \boldsymbol{\Gamma} = \boldsymbol{\dot{\chi}} \boldsymbol{\Phi^{\dagger}}, 
\end{equation}
where \(\boldsymbol{\Phi^{\dagger}}\) is the Moore-Penrose inverse.  A numerically efficient inversion of the augmented data matrix \(\boldsymbol{\Phi}\) can be obtained by singular value decomposition (SVD), such that
\begin{equation}\label{eq:io_space}
    \boldsymbol{\Phi} = \begin{bmatrix}\boldsymbol{U}_q & \boldsymbol{U}_s\end{bmatrix}  \begin{bmatrix} \boldsymbol{\Sigma}_q & \boldsymbol{0} \\ \boldsymbol{0}  & \boldsymbol{\Sigma}_s
\end{bmatrix} \begin{bmatrix}\boldsymbol{V_q} &\boldsymbol{V_s}\end{bmatrix}^*,
\end{equation}
where the operator $^*$ denotes the conjugate transpose, $q\leq N_{\Phi} = \min(n+m,N)$ is a rank truncation integer to be determined, $\boldsymbol {\Sigma}_q \in \mathbb{R}^{q \times q}$ is a diagonal matrix with the highest $q$ number of singular values of $\boldsymbol{\Phi}$, $\boldsymbol {U}_q \in \mathbb{R}^{(n+m) \times q}$, and $\boldsymbol {V}_q^* \in \mathbb{R}^{q \times N}$. If $q$ is selected such that $\boldsymbol{\Sigma}_q$ contains the most dominant singular values of $\boldsymbol{\Phi}$, then it can be concluded that $\boldsymbol{\Phi}\approx \boldsymbol {U}_q \boldsymbol{\Sigma}_q \boldsymbol {V}_q^*$ and \eqref{eq:Gamma_est} yields
\begin{equation}
   \boldsymbol{\Gamma}  \approx \hat{\boldsymbol{\Gamma}} = \boldsymbol{\dot{\chi}} \boldsymbol{V}_q  \boldsymbol{\Sigma}_q^{-1} \boldsymbol{U}_q^*. 
\end{equation}
The corresponding approximations of \( \boldsymbol{A} \) as $\hat{\boldsymbol{A}}$ and \(\boldsymbol{B} \) as $\hat{\boldsymbol{B}}$ can now be extracted from $\hat{\boldsymbol{\Gamma}}$,
\begin{equation}
   \hat{\boldsymbol{\Gamma}} =\boldsymbol{[\hat{A} \quad \hat{B}] }.
\end{equation}

The truncation orders $q$ in \eqref{eq:io_space} is determined using an energy-based threshold criterion, which prescribes the smallest integer $q$ that retains a given minimum cumulative energy percentage $E_\text{min}$,
\begin{equation*}
 \frac{\sum_{i=1}^{q} \sigma_i^2}{\sum_{i=1}^{N_\Phi} \sigma_i^2} \geq E_\text{min},
\end{equation*}
where $\sigma_i$ is the $i^\text{th}$ singular value of $\boldsymbol{\Phi}$. This truncation order thus represents a critical trade-off between numerical efficiency and model fidelity under noisy data.
}

\subsection{Prediction Error Minimization}
To refine the identified model and attain higher fidelity, PEM is employed on the outcome of DMDc to refine the model parameters. A parameterized model structure 
\begin{equation}\label{eq:pem}
  \boldsymbol {\dot{\chi}} = \begin{bmatrix} \boldsymbol{\hat{A}}(\theta) & \boldsymbol{\hat{B}}(\theta) \end{bmatrix} \boldsymbol{\Phi} + \boldsymbol{W},
\end{equation} 
is now defined, where \( \theta \) represents the model parameters to be optimized, and \( \boldsymbol{W} \) accounts for process noise in the observations. The prediction error \( \mathbf{e}(t, \theta)\) is computed using
\begin{equation}\label{eq:pem-error}
\mathbf{e}(t, \theta) = \boldsymbol{\dot{\chi}} - \boldsymbol{\dot{\hat{\chi}}},
\end{equation}
where \( \boldsymbol{\dot{\chi}}\) is the measured derivatives and \( \boldsymbol{\dot{\hat{\chi}}} \) is the model-predicted derivatives defined as
\begin{equation}
  \boldsymbol{\dot{\hat{\chi}}} = \begin{bmatrix} \boldsymbol{\hat{A}}(\theta) & \boldsymbol{\hat{B}}(\theta) \end{bmatrix} \boldsymbol{\Phi}.
\end{equation}
Minimizing the prediction error cost function given by 
\begin{equation} 
  J(\theta) =   \langle \mathbf{e}(t, \theta), \mathbf{e}(t, \theta) \rangle,
\end{equation}
the PEM model estimate is obtained with
\begin{equation}\label{eq:obj}
\hat{\theta} = \arg \min_{\theta} J(\theta). 
\end{equation}
This problem is solved using gradient-based numerical optimization to minimize the energy of the prediction error. {\color{black} Estimates $\boldsymbol{\hat{A}}(\hat{\theta})$ and $\boldsymbol{\hat{B}}(\hat{\theta})$ are then obtained for the system matrices $\boldsymbol{A}$ and $\boldsymbol{B}$, respectively. These matrices provide an identified model of the system to design the optimal DFC by \eqref{DFC ARE} and \eqref{DFC optimal gain} and to compare to the proposed model-free PI solution on the AML system.}

\section{Implementation and Testing on a Magnetic Levitation System}\label{sec:ResultAndDiscussion}
\subsection{Active Magnetic Levitation (AML) System}\label{sec:model}
Figure~\ref{Schematic} illustrates the schematics of the AML test setup used in this study, with two electromagnetic actuators and two magnetic target disks. The electromagnetic actuators consist of two copper coils, with the direction of current determining the polarity of the magnetic fields. The currents in Coil~1 and Coil~2 are represented by \(\mathrm{U}_1\)  and \(\mathrm{U}_2\), respectively. The symbols \(\mathcal{N}\) and \(\mathcal{S}\) denote the north and south poles of the magnetic disks, as well as the polarity of the electromagnetic actuator for positive current input. The mass of the disks is represented by \(\textit{M}\). The distance from Coil 2 to Disk 2 is denoted by \(y_2\), the distance from Coil 1 to Disk 1 by \(y_1\), 
and the distance between the two coils by \(y_c\).

\begin{figure}[ht]
        \centering
\includegraphics[scale=0.75]{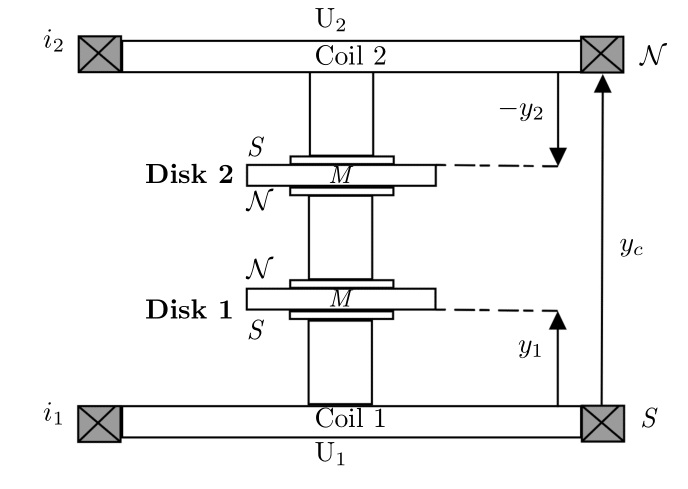}
        \caption{Schematic of the magnetic levitation system.}
        \label{Schematic}
\end{figure}

\begin{table}[ht]
\centering
\caption{System Parameters of the Magnetic Levitation System}
\label{tab:sys_params}
\begin{tabular}{lcc}
\hline
\textbf{Parameter} & \textbf{Symbol} & \textbf{Value} \\
\hline
Magnet Mass & $M$ & $0.126$ kg \\
Gravity & $g$ & $9.81$ m/s$^2$ \\
Damping Coefficient & $c_1,c_2$ & $0.96$ Kg/s \\
Actuator Gain (Inverse) & $a$ & $4.0442 \times 10^4$ A/(N$\cdot$m$^4$) \\
Actuator Offset & $b$ & $0.0591$ m \\
Force Parameter & $c$ & $4.4408 \times 10^{-8}$ N$\cdot$m$^4$ \\
Magnet-Magnet Offset & $d$ & $0.042$ m \\
Coil Distance & $y_c$ & $0.133$ m \\
\hline
\end{tabular}
\end{table}

The equations of motion governing the two-disk AML system are given as
\begin{align}
    \begin{bmatrix} \label{nonlinear equation}
            \ddot{y}_1\\
            \ddot{y}_2   
    \end{bmatrix}
    =\frac{1}{M}
    \begin{bmatrix}
            F_{u11} - F_{u12} - F_{m12} -c_1\dot{y}_1 - Mg\\
            F_{u22} - F_{u21} + F_{m12} - c_2\dot{y}_2 - Mg
    \end{bmatrix}.
\end{align}
where \(c_j\) is the damping coefficient of the magnets, \(g\) represents the acceleration due to gravity, \(F_{m12}\) is the force between Disks 1 and 2, and \(F_{ujk}\) is the force between Coil \(j\) and Disk \(k\) \cite{zaheer2021derivative}. The equations for the forces are given as
\begin{align*}
    F_{u11} &= \frac{\mathrm{U}_1}{a(y_1 + b)^4}, & 
    F_{u12} &= \frac{\mathrm{U}_1}{a(y_c + y_2 + b)^4}, \\
    F_{u21} &= \frac{\mathrm{U}_2}{a(y_c - y_1 + b)^4}, & 
    F_{u22} &= \frac{\mathrm{U}_2}{a(-y_2 + b)^4}, \\
    F_{m12} &= \frac{c}{(y_{12} + d)^4},
\end{align*}
where $a$, $b$, $c$, and $d$ are coefficients determined by the properties of
the magnetic configuration and \(y_{12}=y_c + y_2-y_1\). By neglecting the small forces \(F_{u12}\), \(F_{u21}\), a linearized representation of \eqref{nonlinear equation} can be obtained in terms of the local control efforts \(u_1 = U_1 - u_{10}\)  and \(u_2 = U_2 - u_{20}\) around the bias currents $u_{10}$ and $u_{20}$, and the local displacement \(y_1^* = y_1 - y_{10}\) and \(y_2^* = y_2 - y_{20}\) about the corresponding magnetic equilibrium \(y_{10}\) and \(y_{20}\). 
Consider the state vector \([x_1, x_2, x_3, x_4]^\mathrm{T} = [y_1^*, \dot{y}_1^*, y_2^*, \dot{y}_2^*]^\mathrm{T}\), and the input vector \([u_1, u_2]^\mathrm{T}\).  
The linearized system dynamics around the nominal equilibrium then takes the form
\begin{align*}
    \begin{bmatrix}
        \dot{x}_1 \\
        \dot{x}_2 \\
        \dot{x}_3 \\
        \dot{x}_4
    \end{bmatrix} 
    &\!=\! 
    \begin{bmatrix}
        0 & 1 & 0 & 0 \\
        a_{21} & \frac{-c_1}{M} & a_{23} & 0 \\
        0 & 0 & 0 & 1 \\
        a_{41} & 0 & a_{43} & \frac{-c_2}{M}
    \end{bmatrix}
    \!
    \begin{bmatrix}
        x_1 \\
        x_2 \\
        x_3\\
        x_4
    \end{bmatrix}
    \!+\!
    \begin{bmatrix}
        0 & 0 \\
        \frac{k_{u1}}{M} & 0 \\
        0 & 0 \\
        0 & \frac{k_{u2}}{M}
    \end{bmatrix}
    \begin{bmatrix}
        u_1 \\
        u_2
    \end{bmatrix},\\ 
    \boldsymbol{y} &= 
    \begin{bmatrix}
        1 & 0 & 0 & 0 \\
        0 & 0 & 1 & 0
    \end{bmatrix}
    \!
    \begin{bmatrix}
        x_1,
        x_2,
        x_3,
        x_4
    \end{bmatrix}\T\!,
\end{align*}
where \(a_{21} = (k_1 + k_{12})/M\), \(a_{23} = -k_{12}/M\), \(a_{41} = -k_{12}/M\), 
    \(a_{43} = (k_2 + k_{12})/M\), $y_{12_0}$ is the equilibrium distance between Disks 1 and 2, and    
\begin{align*}
    \begin{aligned}
        k_1 &= \frac{4u_{10}}{a(y_{10} + b)^5}, \quad k_{u1} = \frac{1}{a(y_{10} + b)^4}, \quad
        k_2 = \frac{4u_{20}}{a(y_{20} + b)^5},\\
        k_{u2} &= \frac{1}{a(-y_{20} + b)^4}, \quad k_{12} = \frac{4c}{(y_{12_0} + d)^5}.
    \end{aligned}
\end{align*}
The magnetic equilibrium and bias currents must satisfy the equations
\begin{equation}\label{bias_input}
\begin{aligned}
 u_{10} &= a (y_{10} + b)^4\left(\frac{c}{(y_{12_0} + d)^4} + Mg\right),\\
 u_{20} &= a(y_{20} + b)^4\left(\frac{c}{(y_{12_0} + d)^4} + Mg\right). 
\end{aligned}
\end{equation}
For the desired magnetic equilibrium of \(y_{10}=0.01\,\text{m}\) and \(y_{20}=-0.02\,\text{m}\), the corresponding bias currents are calculated according to system parameters in Table \ref{tab:sys_params}. {\color{black}The resulting state-space matrices for the linearized AML system are 
\begin{equation*}
\setlength{\arraycolsep}{2.5pt}
\boldsymbol{A}\!=\! \begin{bmatrix}
0 & 1 & 0 & 0 \\
567.8 & -7.6 & 0 & 0 \\
0 & 0 & 0 & 1 \\
0 & 0 & 1003.7 & -7.6
\end{bmatrix}\!, ~
\boldsymbol{B}\!=\!\begin{bmatrix}
0 & 0 \\
8.6077 & 0 \\
0 & 0 \\
0 & 83.9636
\end{bmatrix}\!.
\end{equation*}}

As noted in Section~\ref{sec:intro}, the linearized state-space model of the AML system is highly sensitive to the equilibrium states set by the positions $y_{10}$ and $y_{20}$, and the bias currents $u_{10}$ and $u_{20}$. Because of variations in the magnetic force constants and error in the simplified relationship in \eqref{bias_input}, it is expected for the actual AML dynamics to be highly uncertain compared to the above nominal model.

\subsection{Simulation Result} \label{sec:simulation}
{\color{black}The data-driven DFC from Algorithm~\ref{algorithm1} is first evaluated on the AML system through numerical simulations. The simulation environment is implemented in MATLAB/Simulink with a fixed sampling interval of $T_s = 1$\,ms. For the collection of the training data, an exploration noise signal is introduced to the input with frequencies in the range of $[-100, 100]$ rad/s and amplitude of $0.1$. The simulations are conducted in MATLAB 2024 on a MacBook Pro M4 with an M4 processor and 24GB of memory. 

The weighting matrices of the cost function are selected as \(\boldsymbol{Q} = \boldsymbol{I}_4\) and \(\boldsymbol{R} = \text{diag}(1, 2)\), where \(\boldsymbol{I}_4\) is a \(4 \times 4\) identity matrix.
For Algorithm~\ref{algorithm1}, the initial stabilizing gain \(\boldsymbol{K}_1\) is obtained through pole placement based on the nominal model in Section~\ref{sec:model},
\[
\boldsymbol{K}_1 = \begin{bmatrix}
-9.7596 & -0.6122 & -2.8462 & -0.0197 \\
 0.5168 &  0.0038 & -1.6957 & -0.1015
\end{bmatrix}.
\]
The termination thresholds of the algorithm are set at $\bar{\zeta}=10^{-8}$
and $\bar{\eta}=10^{-6}$, and the sample duration $NT$ is 2\,s.}

\begin{figure}
    \centering
    \includegraphics[width=\linewidth]{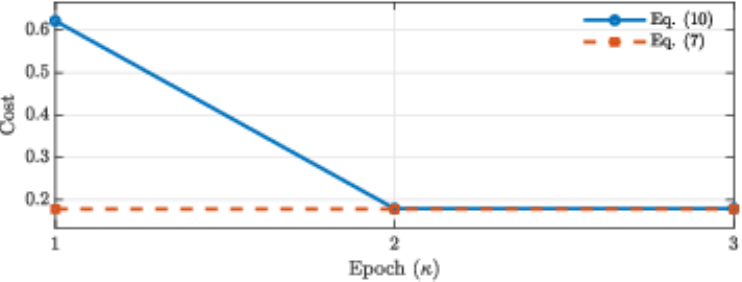}
    \includegraphics[width=\linewidth]{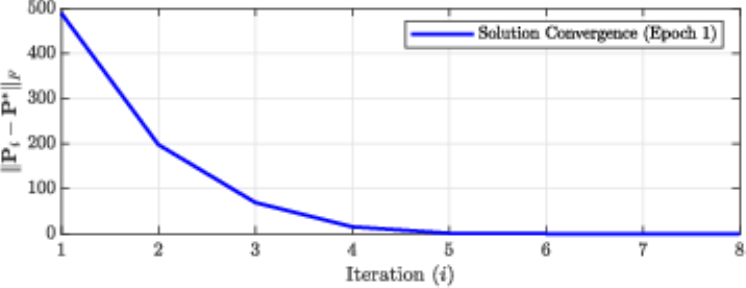}
    \includegraphics[width=\linewidth]{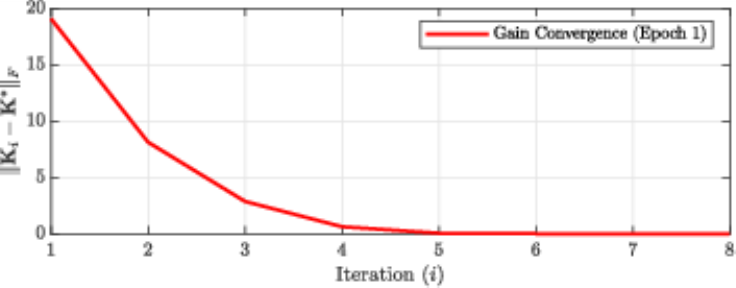}
\caption{\textcolor{black}{Top: The evolution of cost function values computed via \eqref{cost_iteration} and \eqref{finite value DFC cost functiont} over three epochs. Middle: The Frobenius norm of the error between the iterative $\boldsymbol{P}_i$ and the model-based solution \eqref{DFC ARE} over the PI iterations of the first epoch. Bottom: The convergence of the control gain error norm $||\boldsymbol{K}_i - \boldsymbol{K}_\text{ARE}||_F$ over the PI iterations of the first epoch.}}
    \label{Cost_diff_sim}
\end{figure}

\textcolor{black}{
Figure \ref{Cost_diff_sim} presents a convergence analysis of Algorithm~\ref{algorithm1}. The results indicate convergence of both the control gain and the expected cost, as determined by equations \eqref{DFC optimal gain} and \eqref{DFC ARE}. The top plot illustrates the cost computed by \eqref{cost_iteration} from the training data and the optimal cost \eqref{finite value DFC cost functiont} calculated analytically. At the first epoch, the cost associated with the initial control gain $\boldsymbol{K}_1$ is $0.6218$, reflecting suboptimal performance and above the expected minimum of $0.1798$. The middle and bottom graphs of Fig.~\ref{Cost_diff_sim} plot the Forbenius norm $\|\cdot\|_F$ of the errors $\boldsymbol{P}_i-\boldsymbol{P}^*$ and $\boldsymbol{K}_i-\boldsymbol{K}^*$, and show that the solutions converge to their optimal values within eight iterations. The control gain converges to the optimal $\boldsymbol{K}_\text{ARE}$, calculated from the nominal linearized model and \eqref{DFC ARE}-\eqref{DFC optimal gain},
\[
\boldsymbol{K}_{\text{ARE}} = \begin{bmatrix}
-13.1301 & -1.1229 & 0.0004 & 0.0000 \\
-0.0001 & -0.0000 & -4.2980 & -0.7191
\end{bmatrix}\!.
\]
The figure shows that the optimal cost is reached at the end of the first epoch, and no further improvements are observed in subsequent epochs. This is expected for the simulation tests, since there is no unknown external disturbance influencing the training process, unlike what is expected for the experimental tests. The final gain at the end of the third epoch is $\boldsymbol{\mathcal{K}}_3 = \boldsymbol{K}_\text{ARE}$.}

\begin{table}[ht]
\centering
\caption{Computational Efficiency Statistics}
\label{tab:comp_efficiency_combined}
\setlength{\tabcolsep}{8pt} 
\renewcommand{\arraystretch}{1.2} 
\begin{tabular}{ccccc}
\hline
 & \textbf{Total number} & \textbf{Max} & \textbf{Min} & \textbf{Avg} \\
 & \textbf{$\kappa$ or $i$} & \textbf{Time} & \textbf{Time} & \textbf{Time} \\ \hline
\textbf{Epochs}  & 3 & 0.280\,s & 0.244\,s & 0.261\,s \\ \hline
\textbf{Iterations}  & 7 & 4.7\,ms & 2.7\,ms & 3.1\,ms \\ \hline
\end{tabular}
\end{table}

\textcolor{black}{Table \ref{tab:comp_efficiency_combined} presents an evaluation of the computational efficiency of the proposed algorithm for training results presented in Fig.~\ref{Cost_diff_sim}, including the data-collection phase. The results indicate that computational time in the first epoch is higher, at approximately 0.28 seconds, while subsequent epochs require about 0.26 and 0.24 seconds each. Over the first epoch, the initial iteration also takes longer at approximately $4.7\times 10^{-3}$ seconds, but afterwards it decreases to below $3\times 10^{-3}$ seconds.}

\subsection{Experimental Result}
\textcolor{black}{
This section evaluates the proposed model-free DFC Algorithm~\ref{algorithm1} implemented and experimentally tested on the AML system shown in Figure~\ref{Figure_for Magnetic}. The experimental setup consists of an ECP MagLev Model 730 \cite{thomas1999manual} with the control architecture implemented in Simulink Real-Time. The equilibrium state and bias currents were set according to \eqref{bias_input} to $y_{10} = 0.01$\,m, $y_{20}=-0.02$\,m, $u_{10} = 1.1396$\,A and $u_{20} = 0.1168$\,A, respectively. Finally, to mitigate high frequency noise amplification in numerical differentiation, a second-order low-pass filter with a cutoff frequency of $f_c = 2$ Hz is integrated into the control loop.}

\begin{figure}
        \centering
        \includegraphics[height=4.5cm]{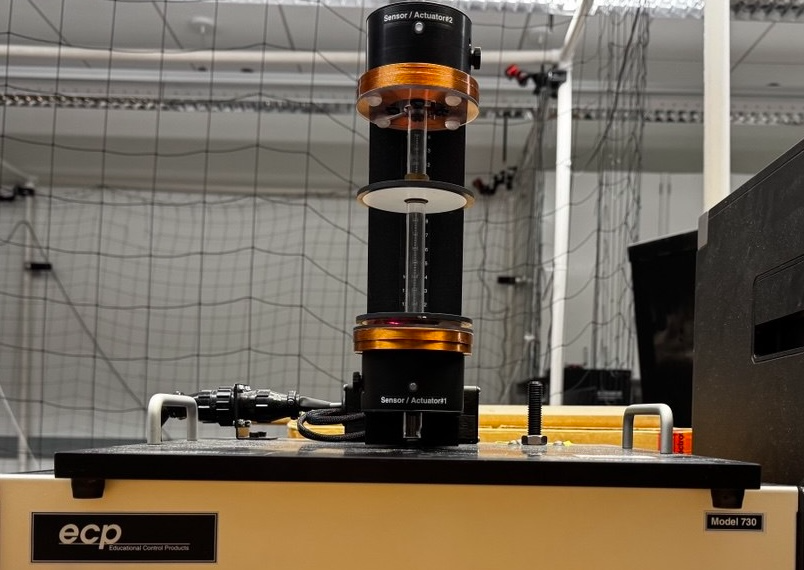}
        \caption{The experimental setup for the two-disk magnetic levitation (MagLev Model 730).}
        \label{Figure_for Magnetic}
\end{figure}

The state response of the AML system from a prescribed initial condition provides the experimental data for the training and evaluation of the model-free DFC in Algorithm~\ref{algorithm1}. Each epoch iteration includes two experiments: one for the training phase and the other for the testing phase. During the training phase, the closed-loop system operating under the initial control gain, generates the necessary training data for the PI process. An excitation signal is introduced at this stage to enrich the frequency content of the training data. Upon completing the PI training of the control law, the epoch transitions to the testing phase, where the trained control law is evaluated and cost function values are recorded.

\textcolor{black}{
The time length of the experimental training data set is set to $4$\,s, compared to the $2$\,s length used in the numerical study of Section~\ref{sec:simulation}. This difference aims to reduce the impact of measurement noise and disturbances observed during the experiments in the training process. All other settings remain unchanged, with weighting matrices \(\boldsymbol{Q} = \boldsymbol{I}_4\) and \(\boldsymbol{R} = \text{diag}(1, 2)\), and sampling time of \(T_s = 1\) ms. The amplitude of the exploratory noise signal is reduced to $0.05$ (from $0.1$ in the simulations) to prevent actuator saturation and ensure system safety, and the epoch tolerance is set to $\bar{\zeta}=0.005$. For the experimental tests, the first epoch is initialized with a stabilizing gain $\boldsymbol{K}_1=\boldsymbol{K}_\text{ARE}$ introduced in Section~\ref{sec:simulation} based on the nominal linearized model.}

\begin{figure}
    \centering
\includegraphics[width=\linewidth]{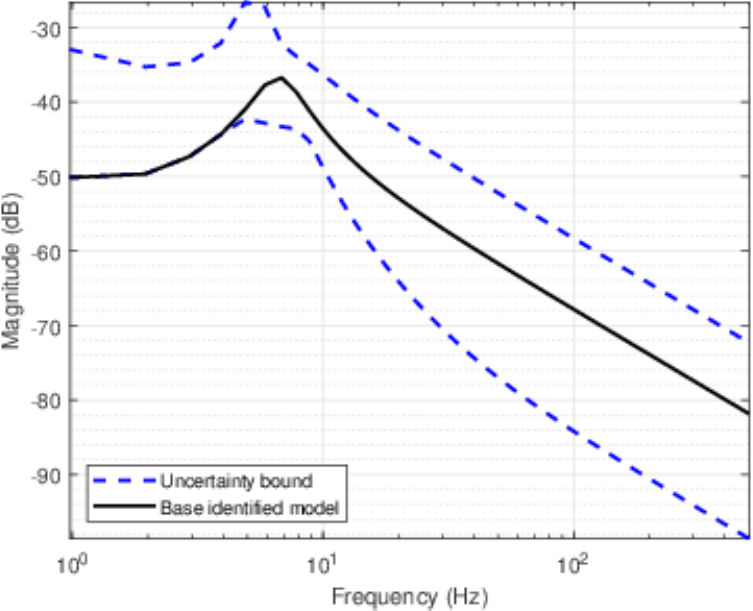}
    \caption{{\color{black}Frequency response of a candidate identified model vs uncertainty range of measured data}}
    \label{fig:magnitude_response}
\end{figure}

\textcolor{black}{In addition to the control law from the model-free PI, the DMDc-PEM method in Sect.~\ref{sec:SID} is implemented to obtain two independently identified models, from which two optimal DFC laws were designed as in \eqref{DFC ARE}-\eqref{DFC optimal gain}. A DMDc truncation order $q=4$ is selected to match the number of states and capture $99\%$ of the cumulative energy $E_\text{min}$, and a carefully selected perturbation noise is used to ensure that the system's modes are adequately excited and the identified model accurately capture the system's dynamics \cite{ghoreyshi2024design}. Figure \ref{fig:magnitude_response} illustrates the magnitude response over frequency of one identified system model from input $u_1$ to output $y_1$. The figure also shows the upper and lower uncertainty bounds determined from the measured frequency response of the experimental setup and with various excitation inputs. The high uncertainty level, particularly at low frequencies, is attributed to the inherent nonlinearity of electromagnetic systems and to errors introduced to the nominal bias current and the magnetic equilibrium. Despite this low-frequency deviation, the identified model demonstrates an acceptable match of the resonance mode and roll-off predicted by the nominal mathematical model. The presence of the uncertainty in the identified model strengthens the argument for the model-free multi-epoch PI technique, as the iterative optimization process allows for further improvement of the trained solutions with convergence guarantees, rather than being limited to a single identification process.}

\begin{figure}
    \centering
    \subfigure{\includegraphics[width=\linewidth]{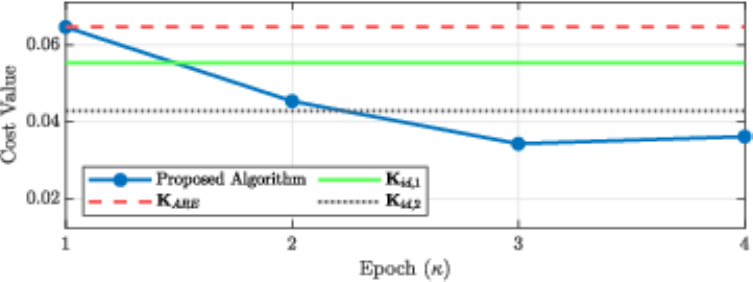}}
    \subfigure{\includegraphics[width=\linewidth]{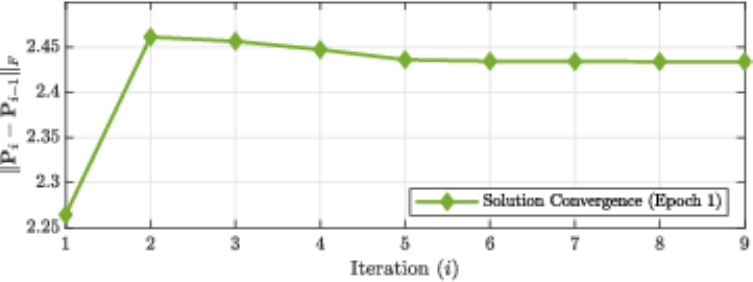}}
    \subfigure{\includegraphics[width=\linewidth]{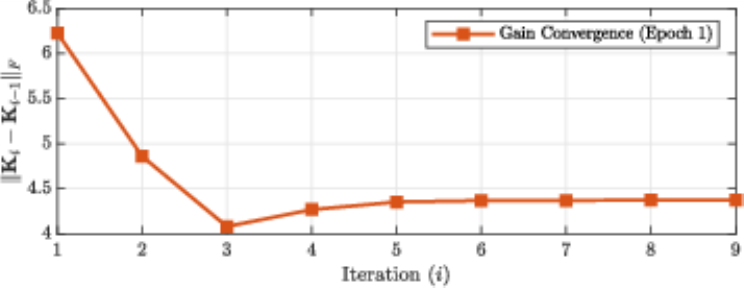}}
    
\caption{\textcolor{black}{Convergence of the Algorithm~\ref{algorithm1} observed for the experimental test. The top graph shows the cost function value over four epochs. A dashed line at $0.0647$ is the nominal cost achieved by $\boldsymbol{K}_\text{ARE}$, which serves as the reference point. The graph also shows the cost values for the solutions of two identified models: $\boldsymbol{K}_{id,1}$ (solid) and $\boldsymbol{K}_{id,2}$ (dotted). The middle and bottom graphs illustrates the convergence of the PI iteration during the first epoch ($\kappa=1$), showing the Frobenius norm of iterative updates for the value matrix ($\|\boldsymbol{P}_i - \boldsymbol{P}_{i-1}\|_F$) and the control gain ($\|\boldsymbol{K}_i - \boldsymbol{K}_{i-1}\|_F$), respectively.}}
\label{Cost_diff_imp} 
\end{figure}

\textcolor{black}{
Figure \ref{Cost_diff_imp} presents the experimental convergence results over four epochs of Algorithm~\ref{algorithm1}. The top graph displays the cost at each epoch, as defined in \eqref{cost_iteration}. The dashed line represents the optimal cost ($0.0647$) predicted by the nominal model and \eqref{DFC ARE}. The optimal DFC laws for the two DMDc-PEM identified system models are given as,
\begin{align}\boldsymbol{K}_\text{id,1} & = \begin{bmatrix} 7.6259 & 0.6401 & 1.1525 & -0.7245 \\ -20.7316 & -1.3916 & -2.8322 & 1.6845 \end{bmatrix},\\
\boldsymbol{K}_\text{id,2} & = \begin{bmatrix} -7.4415 & -0.4775 & -1.7809 & -0.2284 \\ 0.2319  & -0.1073 & -0.5889 & 0.2072 \end{bmatrix},
\end{align}
and their corresponding cost values are also included in the figure. The initial controller $\boldsymbol{\mathcal{K}}_1$ of the first epoch equals the nominal model-based optimal control, and it yields the highest cost among all controllers including those for the identified system models. After the conclusion of first epoch, the cost decreases to $0.0454$, demonstrating that the updated controller outperforms the initial nominal control solution. This controller is also an improvement from $\boldsymbol{K}_\text{id,2}$, which displays a cost of $0.0554$. In Epochs 3 and 4, the cost remains at $0.0343$ and $0.0362$, respectively, indicating robust convergence of the solution. The final cost difference between the last two epochs is $0.0019$, and both solutions outperform $\boldsymbol{K}_\text{id,2}$ with a measured cost of $0.0428$. These findings confirm the hypothesis of using the iteration over epochs to find further improvements by updating the initial control and the training data. In comparison, the DMDc-PEM identification method yields two independent DFC laws with different cost levels and without a systematic procedure to improve the solution. The final control gain of Algorithm~\ref{algorithm1} is $$\boldsymbol{\mathcal{K}}_{4} = \begin{bmatrix} 3.5055 & -1.1141 & -0.1818 & 0.1072 \\ 0.0893 & -0.2294 & 0.1449 & -0.1830 \end{bmatrix}.$$} 

Noise and training bias during a PI iterations of Algorithm~\ref{algorithm1} can significantly affect the performance of the resulting DFC law. While the conditions of persistent excitation partially address this issue, it is often difficult to satisfy and/or verify them in practical applications. By training across multiple epochs, the PI procedure in Algorithm~\ref{algorithm1} can better explore the data space and mitigate the impact of biased samples, leading to improved convergence and more reliable control performance. 

{\color{black} The convergence of the PI solution in Algorithm~\ref{algorithm1} can be further examined in Fig.~\ref{Cost_diff_imp}, which illustrates the Frobenius norm of the error $\|\boldsymbol{P}_i - \boldsymbol{P}_{i-1}\|_F$ and $\|\boldsymbol{K}_i - \boldsymbol{K}_{i-1}\|_F$ over two consecutive PI iterations of Epoch 1. These results validate the theoretical guarantees of convergence for $\boldsymbol{P}_i$ $\boldsymbol{K}_i$ of the PI algorithm.}

\begin{figure}
    \centering
    \subfigure{\includegraphics[width=\linewidth]{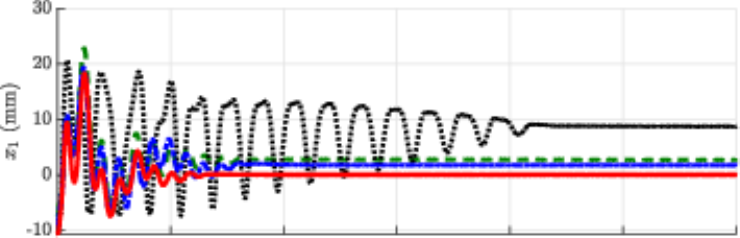}}
    \subfigure{\includegraphics[width=\linewidth]{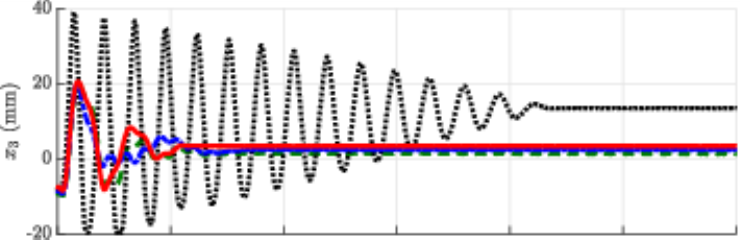}}
    \subfigure{\includegraphics[width=\linewidth]{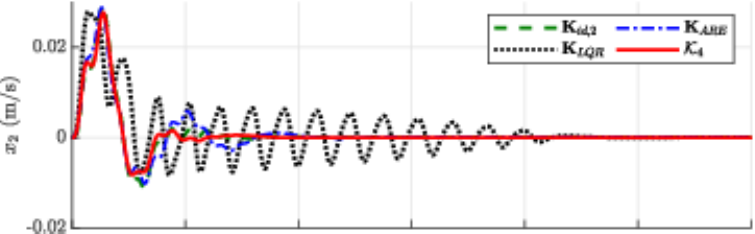}}
    \subfigure{\includegraphics[width=\linewidth]{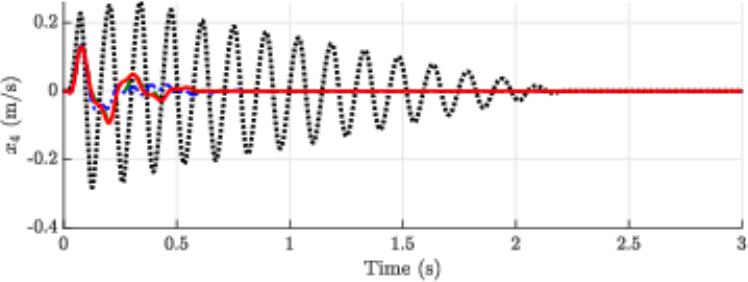}}
\caption{\textcolor{black}{State response comparison of the MagLev system. Trajectories are shown for the proposed learned gain \(\mathbf{\mathcal{K}}_4\), the model-based gain \(\mathbf{K}_\text{ARE}\), the LQR gain \(\mathbf{K}_{LQR}\), and the system identification gain \(\mathbf{K}_{id,2}\).}}
    \label{position_signal_last}
\end{figure}

\textcolor{black}{Figure \ref{position_signal_last} shows the state response of the closed-loop AML system from an fixed initial condition and approaching the equilibrium state. The closed-loop system is tested under three different DFC gains: $\boldsymbol{\mathcal{K}}_4$, $\boldsymbol{K}_\text{ARE}$, and $\boldsymbol{K}_\text{id,2}$. For completeness, the figure also shows the response under a standard state-feedback LQR control law $\boldsymbol{K}_\text{LQR}$. The responses of ${x}_1$ and ${x}_2$ show that the AML system approach the equilibrium position at zero with the DFC gains, while the state feedback control $\boldsymbol{K}_\text{LQR}$ settle the states away the equilibrium states. This highlights the function of the DFC control law in compensating for errors in equilibrium  calculations of $y_{10}$, $y_{20}$, $u_{10}$ and $u_{20}$, and the resulting bias in the feedback measurement. While the DFC can find the actual equilibrium state of the system, the standard state-feedback settle the system away from the equilibrium when the feedback signal is biased. This has important implications in AML systems as to be observed in Fig.~\ref{input_signal_imp_end}.}

\begin{figure}
    \centering
    \subfigure{\includegraphics[scale=.7]{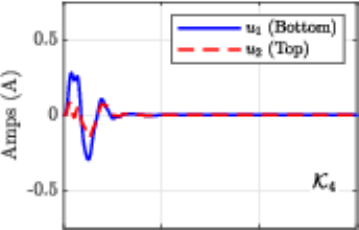}}
    \subfigure{\includegraphics[scale=.7]{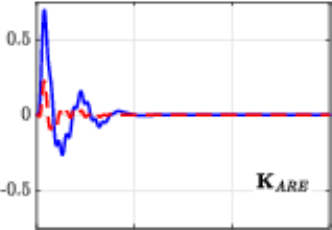}}
    \subfigure{\includegraphics[scale=.7]{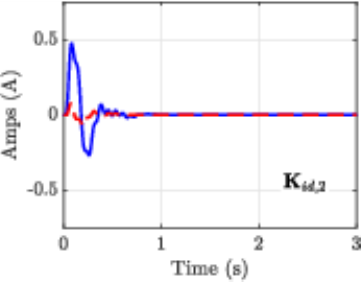}}
    \subfigure{\includegraphics[scale=.7]{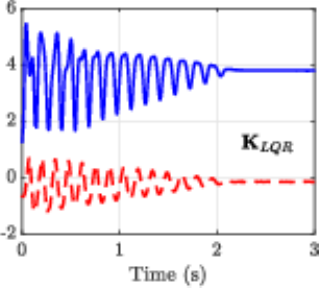}}
\caption{\textcolor{black}{Control signals corresponding to the local inputs applied to the MagLev Model 730 for the proposed model-free learned gain \(\boldsymbol{\mathcal{K}}_4\), the model-based gain \(\boldsymbol{K}_\text{ARE}\), the system identification gain \(\boldsymbol{K}_{id,2}\), and the standard LQR baseline.}}
    \label{input_signal_imp_end}
\end{figure}

\textcolor{black}{Figure \ref{input_signal_imp_end} illustrates the inputs $u_1$ and $u_2$ of the lower actuator coil and the upper coil, respectively, for the state response test presented in Fig.~\ref{position_signal_last}. The figure demonstrates that the state-feedback $\boldsymbol{K}_\text{LQR}$ maintains $u_1$ near 4\,A at the steady state, approaching the saturation limit of the AML power supply. In practical applications, this represents a significant reduction in the control authority of the closed-loop system to compensate for external disturbances. On the other hand, DFC gains display control signals that approach zero, as expected at the equilibrium state. Among the DFC laws, the trained gain $\boldsymbol{\mathcal{K}}_4$ displays the smallest input energy, with the maximum amplitude of the input limited to approximately $\pm 0.2$\,A.}

\section{Conclusion} \label{sec:Conclusion} 
{\color{black}A comprehensive comparison study of data-driven optimal DFC methods was presented for an AML system. The DFC approach was selected for this study to address the uncertainty in the magnetic equilibrium of the system and the resulting bias in the system output, whereas the data-driven approach aimed to compensate for dynamics uncertainty in the nominal model. An enhanced model-free PI framework was introduced as part of this study, which incorporates an epoch iteration loop to mitigate the influence of measurement noise and enhance the robustness of the PI process. This approach was compared to an alternative indirect method, in which a DMDc-PEM model identification technique was introduced to determine the system dynamics and derive the DFC law. The two data-driven DFC solutions was compared to benchmark solutions derived from a nominal model through numerical simulation and experimental tests. 

Simulations and experimental results demonstrated that both data-driven DFC solutions outperform the benchmarch DFC derived from the nominal model in the presence of model uncertainty, resulting in a lower cost value during the evaluation tests. Among the data-driven solutions, the solution of the enhanced PI algorithm yielded a lower cost than the indirect alternative, benefited by the theoretical guarantees of improvement over iterations shown for the PI method. The study also demonstrated the advantages of DFC over conventional state-feedback control in applications with uncertainty in the equilibrium states and feedback measurement bias.

Future work will aim to improve the scalability of the data-driven solutions to AML systems with higher order dynamics. Although it was demonstrated that the computational complexity of the PI algorithm increases linearly with total iteration and epoch numbers, the increase in complexity is cubic relative to the order of the controlled system, and solutions to address this increase will be investigated. Data-driven controller design with guaranteed robustness margins is another area of future work to enable applications in complex AML systems.
}

\bibliography{references}
\bibliographystyle{IEEEtranS}

\end{document}